\documentclass[aps,prl,twocolumn,superscriptaddress,groupedaddress]{revtex4} 

\usepackage{natbib}
\usepackage{amssymb,amsmath,latexsym}
\usepackage{amsthm}
\usepackage{hyper ref}
\usepackage{bm}
\usepackage{microtype}
\usepackage{url}

\usepackage[lined,ruled,vlined]{algorithm2e}

\usepackage{xcolor}
\theoremstyle{plain}

\newtheorem{theorem}{Theorem}

\theoremstyle{definition}
\newtheorem{definition}{Definition}

\theoremstyle{remark}

\newcommand{\abs}[1]{\left| #1 \right|}

\newcommand{\ket}[1]{\left| #1 \right\rangle}
\newcommand{\proj}[1]{| #1 \rangle\!\langle #1 |}
\newcommand{\pname}[1]{\textsf{#1}}
\newcommand{\eps}{\varepsilon}
\newcommand{\I}{\mathbb{I}}
\newcommand{\qsucc}{q_\mathrm{succ}}

\newcommand{\suppl}{appendix}

\begin{document}

\title{Reference frame agreement in quantum networks}

\author{Tanvirul Islam} \email[]{tanvir@locc.la}
\affiliation{Centre for Quantum Technologies, National University of Singapore, 3 Science Drive 2, 117543 Singapore}
\affiliation{School of Computing, National University of Singapore, 13 Computing Drive, 117417 Singapore}
\author{Lo\"ick Magnin} \email[]{loick@locc.la}
\affiliation{Centre for Quantum Technologies, National University of Singapore, 3 Science Drive 2, 117543 Singapore}
\author{Brandon Sorg}
\affiliation{Centre for Quantum Technologies, National University of Singapore, 3 Science Drive 2, 117543 Singapore}
\author{Stephanie Wehner} \email[]{steph@locc.la}
\affiliation{Centre for Quantum Technologies, National University of Singapore, 3 Science Drive 2, 117543 Singapore}
\affiliation{School of Computing, National University of Singapore, 13 Computing Drive, 117417 Singapore}

\begin{abstract}
In order to communicate information in a quantum network effectively, all network nodes should share a common reference frame. Here, we propose to study how well $m$ nodes in a quantum network can establish a common reference frame from scratch, even though $t$ of them may be arbitrarily faulty. We present a protocol that allows all correctly functioning nodes to agree on a common reference frame as long as not more than $t < m/3$ nodes are faulty. Our protocol furthermore has the appealing property that it allows any existing two-party protocol for reference frame agreement to be lifted to a protocol for a quantum network.
\end{abstract}

\maketitle

Quantum networks are
gaining importance~\cite{kim08} for a variety of tasks such as
quantum distributed computing~\cite{BBG+13}, quantum cloud computing~\cite{BKB+12} and quantum key distribution (see e.g.~\cite{Elt02, PPM08, SLB+11, SFI+11}). 
From the current architecture of the internet one can predict that any such network will contain a large number of nodes that are distributed over widespread geographical locations on earth 
or on satellites~\cite{BTD+09,PYB+05,AFJ+08,BAM+06,AJP+03} and connected via quantum and classical communication channels~\cite{CZKM}. 
Some of the many challenges in building a quantum network spanning long distances are the ability to perform
quantum error correction~\cite{Shor95} and construction of quantum repeaters~ (see e.g.~\cite{SSH+11}). 
Yet, before we can hope to implement even such basic building blocks effectively, we would like all nodes in the quantum network to agree on a common reference frame to enable easy quantum communication.

A significant research effort has been devoted to developing protocols for agreeing on a reference frame between just two nodes~\cite{MP95, PS01, BBM04, CD04, BM06, GLM06, BRS07, SG12}. 
Such protocols demand quantum communication because in the absence of a pre-shared reference frame, a node cannot meaningfully share directional information to a distant node by exchanging only classical data. Instead, a quantum system must be sent, for example a qubit with its Bloch vector pointing in the required direction. 
A simple two-node protocol is thus
to send many copies of the same qubit such that the receiver can approximate the direction with certain level of accuracy.

Here, our goal is to allow $m > 2$ number of nodes in a quantum network to agree on a common reference frame, where in this first work we assume a fully connected network graph. That is, every node is connected to every other node using both classical and quantum communication channels. Why is this problem any more difficult than solving the problem for two nodes? 
Note that in an ideal case, where all the nodes are perfect and the channels connecting them are error-free, one node can send a reference frame to everyone else, 
and everyone can subsequently use that as their common frame of reference. But one can see that in a practical network, where some of the nodes can be arbitrarily faulty this simple method will not work because if the sending node is faulty, then it might send a different frame to different receivers and thus cause different nodes to output different reference frames.
That is, it can prevent them from \emph{agreeing} on a \emph{common} frame. Dealing with faulty nodes in a quantum network is challenging because we do not know \textit{apriori} which nodes are faulty, and to make the things even worse, the faulty nodes might have correlated errors. This is quite realistic in a practical setting where for example their hardware might have the same manufacturing defects, they might be located at a geographical location which is going through some disaster, or they might even be hijacked by an adversary trying to disrupt the network. 
Such arbitrarily correlated errors can all be characterized by imagining a worst case scenario in which the $t$ faulty nodes in the network are indeed actively cooperating to thwart our 
efforts in trying to establish a common reference frame. 

To state the requirements for our protocol for establishing a common Cartesian reference frame, let us first clarify what it means to (approximately) agree on a frame.
Let $v_i = (\alpha_i,\beta_i,\gamma_i)$ be the classical representation of the 
vector $\alpha_i \vec{\bm{x}}_i + \beta_i \vec{\bm{y}}_i + \gamma_i \vec{\bm{z}}_i$ held by the node $P_i$, expressed with relative to its local
Cartesian frame $(\vec{\bm{x}}_i,\vec{\bm{y}}_i,\vec{\bm{z}}_i)$.
We denote $d(v_i,v_j)$ the Euclidean distance between the two vectors~\footnote{For unit vectors $d$ takes values between $0$ and $2$.}, expressed with respect to the same reference frame. That is, when considering the distance between vectors $v_i$ held by node
$P_i$ and $v_j$ held by node $P_j$, we translate them into one fixed frame which without loss of generality we take to be the frame of the first node $P_i$. 
Informally, $P_i$ and $P_j$ thus (approximately) $\eta$-agree on a reference frame if $d(v_i,v_j) \leq \eta$ where $\eta$ is ideally small. We are now ready to define our goal.
 \begin{definition}
	For $\eta > 0$, a $\eta$-reference frame consensus protocol among $m$ network nodes is a protocol such that 
\begin{description}
	\item[Termination] Each correct node $P_i$ terminates the protocol, and outputs a reference frame $v_i$.
	\item[Consistency] For all pairs of correct nodes $P_i$ and $P_j$ we have $d(v_i,v_j) \leq \eta$. 
\end{description}
\end{definition}
Note that consistency does not require that all the correct nodes share the same reference frame ($\eta = 0$), 
but that each node has an approximation of it ($\eta$ is small). This is important because already any two-node
protocol using only a finite number of rounds of communication cannot allow the two nodes to share a frame exactly.

\section{Results}
\label{sec:results}

We introduce the first protocol to solve the reference frame agreement problem in a quantum network of $m$ nodes of which $t < m/3$ can be arbitrarily faulty. 
Our protocol has the appealing feature that it can use any two-node protocol as a black box. Such two-node protocols~\cite{BRS07} are characterized by the accuracy $\delta$ (i.e., the two nodes $\delta$-agree) and the success probability $q_{\rm succ}$ with which such an approximation guarantee is achieved. 

\begin{theorem}
	Given any two-node protocol to estimate a direction with accuracy $\delta$ and success probability $q_{\rm succ}$, the protocol \pname{RF-Consensus} 
	is a $(30\delta)$-reference frame consensus protocol tolerant to $t < m/3$ faulty nodes. It succeeds with probability at least $\qsucc^{m^2}$.
\end{theorem}

Our protocol is \emph{efficient} as we need only a linear (in the number of nodes $m$) number of rounds of quantum communication. 
As an example, we take the simplest two-node protocol in which the sender encodes the direction in the Bloch vector of a qubit and sends $n$ identical copies of it to the receiver. For accuracy $\delta>0$, the success probability of this two-node protocol is $\qsucc \geq 1 - e^{\Omega(-n\delta^2)}$. From this, we get the overall success probability of our protocol to be $q_{\rm succ}^{m^2}\geq 1-e^{-\Omega(n\delta^2-\log m)}$. We also show that this setting is \emph{robust} to noise on the channel connecting any two nodes. To give some examples of parameters, protocol \pname{RF-Consensus} achieves accuracy $30\delta = 0.02$ with success probability 99\% in a network of $m=10$ nodes with noiseless communication, if each node transmits $n\approx 3.1\times10^8$ qubits at each round.

Our protocol uses ideas of~\cite{FM00} which solves a simpler problem from classical distributed computing called Byzantine agreement~\cite{LSP82}, in particular we use classical consensus as a subroutine. This problem has been extensively 
studied using synchronous~\cite{FM97,BPV06} and asynchronous~\cite{AAM10,ADH08,Bra84,CT93} classical communication, as well as quantum communication~\cite{BH05}, also in a fail-stop model in which the faulty nodes can prevent the protocol from ever terminating~\cite{FGM01}. 
There, the correct nodes should 
perfectly agree on a single classical bit. Recall that we cannot send a direction classically without a shared reference frame, and hence we cannot use such protocols. 
In addition, we face two extra challenges: 
First, we are dealing with a continuous set of outcomes; And second, it is impossible to transmit a direction perfectly using a finite amount of communication, 
even on an otherwise perfect channel.
In quantum networks, furthermore, we also have errors on the communication
channel, which are pretty much unavoidable in a regime where we cannot easily perform quantum error correction due to the lack of a common frame.
In the Byzantine problem such errors would be attributed to faulty nodes, but in our setting this would mean that \emph{all} nodes in the network are faulty and no protocol could 
ever hope to succeed. Here, we thus require a careful treatment of such approximation errors.

\section{Model of communication}
\label{sec:model}

We assume that all the communication channels are public (faulty nodes can adapt their strategy depending on the network traffic), authenticated (faulty nodes cannot tamper with the channel connecting correct nodes), and synchronous (correct nodes know when they are supposed to receive a message, and if none is received, e.g. due to communication error, the protocol continues which ensures that our protocol cannot stall indefinitely). 

We only use quantum communications to send a direction between a sender and a receiver. As an example we use protocol \pname{2ED}, one of the simplest possible protocols: a sender creates many identical qubits with their Bloch vector pointing to the intended direction and the receiver measures them with Pauli measurements. From the statistics of the measurement outcomes, the receiver then estimates the Bloch vector's direction closely with high success probability. We use this protocol since it has some experimental  advantages for implementation: it does not require any quantum memory or creation of entangled states, and it succeeds even if the quantum channel has a depolarizing noise. But the downside of this choice is that our protocol is not optimal in the number of qubits sent to achieve a certain accuracy. Optimal protocols can align frames in the so-called Heisenberg limit, they have a quadratic gain over the one we use here~\cite{GLM06}. 

\begin{algorithm}
\SetAlgorithmName{Protocol}{protocol}{List of Protocols}
	\LinesNumbered
	\DontPrintSemicolon
	\SetKwInOut{Input}{input}\SetKwInOut{Output}{output}
	\Input{Sender, direction $u$}
	\Output{Receiver, direction $v$ }

\SetKwBlock{SEND}{Sender: \pname{2ED-Send}}{}	
\SEND{
	Prepare $3n$ qubits with direction $u$ \; 
	Send them to the receiver

}
\SetKwBlock{RECIEVE}{Receiver: \pname{2ED-Receive}}{}	
\RECIEVE{
	Receive $3n$ qubits from the sender\;
	Measure $n$ qubits with $\sigma_x$ and compute $p_x$, the frequency of getting outcome~$+1$ \;
	Similarly on the remaining qubits, compute $p_y$ and $p_z$ with measurements $\sigma_y$ and $\sigma_z$ on $n$ qubits each \;
	Assign $x \leftarrow 2 p_x-1$, $y \leftarrow 2 p_y-1$, $z \leftarrow 2 p_z-1$;\ 
	Assign $l \leftarrow \sqrt{x^2+y^2+z^2}$\;
	Output $v \leftarrow (x/l,y/l,z/l)$\;
}
\caption{\pname{2ED}}
\end{algorithm}

We prove the following theorem in the \suppl{}.
\begin{theorem}
For all $\delta>0$, using a depolarizing channel $\rho \mapsto (1-\eps)\rho + \eps\I/2$ between the sender and the receiver, protocol \pname{2ED} provides to the receiver a $(1-\eps)\delta +\frac{5\eps}{2}$ approximation of the sender's direction. It succeeds with probability $\qsucc \geq 1-e^{-\Omega(\delta^2n)}$.
\end{theorem}

\section{Protocols  \label{sec:protocols}}
In this Section, we present a summary of our protocols and an outline of their proof of correctness. For further detail, we refer to the \suppl.

Our protocol works in two phases: First, a node is elected as the \emph{king} $P_k$. Second, the king choses a direction $w_k$ and sends it to all the other nodes. We denote $w_i$ the direction received by the node $P_i$ in its own frame. If the king is not faulty, \pname{2ED} ensures that $d(w_i,w_k)\leq\delta$. Then the correct nodes should decide either all to accept this direction (they output $v_i \approx w_k$ in their respective own frame), or all to reject it (output $\perp$). This second phase is known as \emph{king consensus}. More formally, a king consensus protocol should satisfy two properties: \emph{$\delta$-persistency}: if the king is not faulty, all the correct nodes $P_i$, should output $v_i$ such that $d(v_i,w_k)\leq\delta$; and \mbox{\emph{$\eta$-consistency}}: All the correct nodes reach a consensus, that is, they either all output $\perp$, or they all output directions that are $\eta$-close to each other, i.e., for all correct nodes $P_i$ and $P_j$, the distance $d(v_i,v_j)\leq\eta$.

We repeat those two phases with different kings as long as a consensus is not reached.
In particular, the protocol will terminate after at most $t+1$ rounds since there are a most $t$ faulty nodes.
\begin{algorithm} 
	\LinesNumbered
	\SetAlgorithmName{Protocol}{protocol}{List of Protocols}
	\DontPrintSemicolon
	\SetKwInOut{Input}{Input}
	\SetKwInOut{Output}{Output}
	\Input{None}
	\Output{A direction $v_i$}
	
	\For{ $k =1$ to $t+1$}{
		$v_i$ = \pname{King-Consensus}($P_k$) \;
		\If{$v_i \neq \perp$}{Output $v_i$}
	}
	\caption{\pname{RF-Consensus}}
	\label{RFC}
\end{algorithm}

The rest of this Letter is thus devoted to construct a king consensus protocol, which is done in three steps.
 
\paragraph{Step 1: Weak Consensus}

We first create a weaker protocol than king consensus by relaxing the condition that the correct nodes either \emph{all} output a direction, or \emph{all} output $\perp$. In a weak consensus, \emph{some} nodes can output $\perp$ and the other a direction. However we keep the condition that if two correct nodes $P_i$ and $P_j$ output directions $u_i$ and $u_j$, they should be close to each other.  Formally, we define a \emph{weak consensus} protocol as a protocol with the following two properties:  \emph{$\delta$-weak persistency}: if there exists a direction $w_k$ such that for every correct node $P_i$, $d(w_i,w_k) \leq \delta$, then  $d(u_i,w_k) \leq \delta$; and \emph{$\eta$-weak consistency}: For every pair of correct nodes $P_i$ and $P_j$ which output $u_i\neq \perp$ and $u_j\neq  \perp$ respectively, we have $d(u_i,u_j)\leq \eta$.

\begin{algorithm}
\SetAlgorithmName{Protocol}{protocol}{List of Protocols}
	\DontPrintSemicolon
	\LinesNumbered
	\SetKwInOut{Input}{Input}
	\SetKwInOut{Output}{Output}
	\Input{Direction $w_i$}
	\Output{Direction $u_i$ or $\perp$}

Send $w_i$ to all other nodes \;
Receive $a_i[j] \leftarrow$ direction received from $P_j$ \;
Create the set $S_i \leftarrow \{ P_j: d(w_i, a_i[j]) \leq 3\delta \}$ \;
\uIf{$|S_i| \geq m-t$}
	{Assign $u_i\leftarrow w_i$}
\Else{Assign $u_i \leftarrow \perp$}
Output $u_i$
\caption{\pname{Weak-Consensus}}
\end{algorithm}

Protocol \pname{Weak-Consensus} achieves $\delta$-weak persistency and $(8\delta)$-weak consistency with probability at least $\qsucc^{m^2-m}$ where $\delta$ is the accuracy achieved with probability $\qsucc$  by the two-party protocol used to send directions.

Here, with probability at least $\qsucc^{m^2-m}$, for every correct nodes $P_i$ and $P_j$, $d(a_i[j], w_j)\leq\delta$. It is easy to see that this protocol is $\delta$-weak persistent. We sketch the proof of the weak consistency. Consider the sets $S_i$ and $S_j$ of two correct nodes $P_i$ and $P_j$. If $u_i\neq\perp$ and $u_j\neq\perp$, then $S_i$ and $S_j$ contains at least one correct node in common, let us call it $P_\alpha$. Thus,
$d(u_i,u_j) \leq d(u_i,a_i[\alpha]) + d(a_i[\alpha],w_\alpha) + d(w_\alpha,a_j[\alpha]) + d(a_j[\alpha],u_j) \leq 3\delta + \delta + \delta + 3\delta = 8\delta$.

\paragraph{Step 2: Graded Consensus.}
In a king consensus protocol, the correct nodes should have a ``global'' behaviour, as they should all either output a direction or $\perp$, whereas in the weak consensus each node has a ``local'' strategy.  A \emph{graded consensus} protocol behaves intermediately. Alongside a direction $v_i \neq \perp$ the nodes also output a grade $g_i \in\{0, 1\}$ which carries a ``global''  property, namely, \emph{$\eta$-graded consistency}: If \emph{any} correct node outputs a grade 1, then the directions between \emph{all} the correct nodes should be $\eta$-close to each other, that is, for every pair ($P_i, P_j$) of correct nodes, $d(v_i,v_j) \leq \eta$.

\begin{algorithm}
\SetAlgorithmName{Protocol}{protocol}{List of Protocols}
	\LinesNumbered
	\DontPrintSemicolon
	\SetKwInOut{Input}{Input}\SetKwInOut{Output}{Output}
	\Input{A direction $w_i$}
	\Output{A direction $v_i$ and a grade $g_i \in \{0,1\}$}
	
	Run \pname{Weak-Consensus}($w_i$)\; 
	\tcp{This initialises the variables $u_i$ and $a_i[j]$'s}
	\uIf{$u_i = \perp$}{Send flag $f_i = 0$ to all other nodes}
	\Else{Send flag $f_i = 1$ to all other nodes}
	\ForAll{nodes $P_j$}{$f_i[j] \leftarrow$ Receive $f_j$}
	\ForAll{nodes $P_j$ with $f_i[j]=1$}{Create set $T_i[j] \leftarrow \{P_k: f_i[k]=1,$ and $d(a_i[j], a_i[k]) \leq 10\delta\}$ }
	Assign $l_i\leftarrow \arg \max \{|T_i[j]|\}$ \;
	\uIf{$f_i  = 1$}{Assign $v_i \leftarrow w_i$}
	\Else{Assign $v_i \leftarrow a_i[l_i]$}
	\uIf{$|T_i[l_i]| > m - t $}{Assign $g_i \leftarrow 1$}
	\Else{Assign $g_i \leftarrow 0$}
	Output $(v_i, g_i)$

\caption{\pname{Graded-Consensus}}
\end{algorithm}

Protocol \pname{Graded-Consensus} achieves $(30\delta)$-graded consistency. It succeeds with probability at least $\qsucc^{m^2-m}$. 
 
The main idea of \pname{Graded-Consensus} is that the nodes which output $\perp$ in the weak consensus inform the other nodes (by sending the flags $f_i$'s). The first consequence is that for all correct nodes $P_\alpha$ and $P_\beta$ with $f_\alpha=f_\beta=1$,  $d(u_\alpha,u_\beta)\leq 8\delta$. The second consequence is that if a correct node has grade 1, then for all correct nodes $P_i$ and $P_j$, the sets $T_i$ and $T_j$ each contains at least one correct node, let us denote them $P_\alpha$ and $P_\beta$. Thus, $d(v_i,u_\alpha) \leq d(v_i,a_i[\alpha]) + d(a_i[\alpha],u_\alpha) \leq 10\delta+\delta = 11\delta$. Finally, we get, $d(v_i,v_j) \leq d(v_i,u_k) + d(u_k,u_l) + d(u_l,v_j) \leq  11\delta + 8\delta + 11\delta = 30\delta$.

\paragraph{Step 3: King Consensus.}
We are ready to present the \pname{King-Consensus} protocol that achieves $\delta$-persistency and $(30\delta)$-consistency. Our protocol uses \pname{Classical-Consensus} as a subroutine. It solves a problem which is closely related to Byzantine agreement. Here, every node $P_i$ starts with a bit $g_i$ and outputs a bit $y_i$. All the correct nodes agree on a bit $b$, that is if $P_i$ is correct, $y_i=b$ where at least one of the correct nodes, $P_j$ has input $g_j=b$. Classical consensus can be reached if there are $t < m/3$ faulty nodes, for an example of such protocol, see e.g.~\cite{PSL80}.

\begin{algorithm}
\LinesNumbered
\SetAlgorithmName{Protocol}{protocol}{List of Protocols}
	\DontPrintSemicolon
	\SetKwInOut{Input}{Input}
	\SetKwInOut{Output}{Output}
	\Input{Id of the king $P_k$.}
	\Output{A direction $v_i$ or $\perp$}

\uIf{I am the king}{
Fix an arbitrary direction $w_k$ \;
Send $w_k$  to all other nodes}
\Else{Receive $w_i \leftarrow$ direction received from the king}
Assign $(v_i, g_i) \leftarrow \pname{Graded-Consensus}(w_i)$\;
Assign $y_i \leftarrow \pname{Classical-Consensus}(g_i)$\; 
\uIf{$y_i = 1$}{Output $v_i$}
\Else{Output $\perp$}
\caption{\pname{King-Consensus}}
\end{algorithm}
If the king is not faulty, then all the correct nodes will have grade $g_i =1$. Hence the classical consensus will also be reached with value $y_i = 1$. So, all the correct nodes will accept the direction shared by the king. If the king is faulty and yet the correct nodes reach a consensus with $y_i = 1$, it means that at least one correct node had grade $1$. In this case the $(30\delta)$-graded consistency implies that $d(v_i,v_j)\leq 30\delta$ for all the correct nodes $P_i$ and $P_j$. As a consequence, \pname{King-Consensus} is ($30\delta$)-consistent, and so is \pname{RF-Consensus}.

\section{Discussion}

We have presented the first protocol for reference frame agreement in a quantum network. Even in the classical setting, the algorithms to solve the Byzantine agreement problem are surprisingly complicated. We would be very keen to know if simpler and more efficient protocols could be designed for our setting, possibly by using entangled states.
It is an interesting open question to construct protocols that also work in an asynchronous communication
model. The latter is already challenging for the classical case~\cite{Bra84,CT93,ADH08,AAM10}, 
so we expect a similar behavior to hold here.  
Another interesting question is whether more faulty nodes than $t<m/3$ can be tolerated. 
If our protocol were to succeed with probability 1 and $\eta$ sufficiently small, we can prove that it is optimal in that sense by adapting the classical proof~\cite{FLM85} to our setting. However, for aligning reference frames, any protocol can only succeed with probability strictly less than 1. This problem has been partially studied in the classical case~\cite{GY89}. Even in the constant error scenario the optimal number of faulty nodes that can be tolerated is not known for the classical Byzantine agreement problem~\cite{FWW06}. This leaves hope to find protocols that can tolerate $t<m/2$ faulty nodes when allowing constant success probability both for Byzantine and reference frame agreement.

\begin{acknowledgments}
We thank Esther H{\"a}nggi and J{\"u}rg Wullschleger for useful discussions.
This work is funded by the Ministry of Education (MOE) and National Research
Foundation Singapore, as well as MOE Tier 3 Grant 
MOE2012-T3-1-009.
\end{acknowledgments}

\bibliography{Byzantine}

\clearpage

\appendix
\section{Appendix}

\section{Estimating Directions}
In this Section, we analyse the protocol \pname{2ED} to exchange a direction between two parties. Since this cannot be done perfectly, the receiver has to \emph{estimate} the direction sent by the sender. This task is formally defined by:
\begin{definition} 
	A $\delta$-\emph{estimate direction} protocol is a two-party protocol where one node (the sender) sends a direction $u$ to the other node (the receiver). Upon termination the receiver gets a $\delta$-approximation $v$ of $u$, that is, $d(u,v) \leq \delta$.
\end{definition}

This simple protocol has several advantages: it does not require any quantum memory or the creation of entangled states, and it succeeds even if the quantum channel has a depolarizing noise.  But the downside of this choice is that the protocol is not optimal in the number of qubits sent to achieve a certain accuracy. Any other protocol can be used here~\cite{BRS07};

\begin{algorithm}
\SetAlgorithmName{Protocol}{protocol}{List of Protocols}
	\LinesNumbered
	\DontPrintSemicolon
	\SetKwInOut{Input}{input}\SetKwInOut{Output}{output}
	\Input{Sender, direction $u$}
	\Output{Receiver, direction $v$ }

\SetKwBlock{SEND}{Sender: \pname{2ED-Send}}{}	
\SEND{
	Prepare $3n$ qubits with direction $u$ \;
	Send them to the receiver

}
\SetKwBlock{RECIEVE}{Receiver: \pname{2ED-Receive}}{}	
\RECIEVE{
	Receive $3n$ qubits from the sender\;
	Measure $n$ qubits with $\sigma_x$ and compute $p_x$, the frequency of getting outcome~$+1$ \;
	Similarly on the remaining qubits, compute $p_y$ and $p_z$ with measurements $\sigma_y$ and $\sigma_z$ on $n$ qubits each \;
	Assign $x \leftarrow 2 p_x-1$, $y \leftarrow 2 p_y-1$, $z \leftarrow 2 p_z-1$;\ 
	Assign $l \leftarrow \sqrt{x^2+y^2+z^2}$\;
	Output $v \leftarrow (x/l,y/l,z/l)$\;
}
\caption{\pname{2ED}}
\end{algorithm}

\setcounter{theorem}{1}
\begin{theorem}
For all $\delta>0$, using a depolarising channel $\rho \mapsto (1-\eps)\rho + \eps\I/2$ between the sender and the receiver, protocol \pname{2ED} provides to the receiver a $(1-\eps)\delta +\frac{5\eps}{2}$ approximation of the sender's direction. It succeeds with probability $\qsucc \geq \left(1-2e^{\left(-2 n\delta^2/25\right)}\right)^3$.
\end{theorem}

\begin{proof}
We will prove this theorem in two steps. First, we consider the case when the communication channel is noise free ($\eps = 0$), and then, we see how depolarizing noise affects the approximation factor. 

In the noise-free case, let us fix $\delta >0$ and denote by $\theta_x, \theta_y$, and $\theta_z$ the angles between $u$ and the $x$-, $y$-, and $z$-axis of the local frame of the receiver. So, $\cos^2\frac{\theta_x}{2}$ is the probability of getting outcome $+1$ after the Pauli measurement $\sigma_x$ on a qubit. Similarly, $\cos^2\frac{\theta_y}{2}$ and $\cos^2\frac{\theta_z}{2}$ are the probabilities for outcome $+1$ on measurement $\sigma_y$ and $\sigma_z$ respectively. 

Now, we will show that each of the following three conditions:
\begin{align}
	\label{eq:px}|p_x-\cos^2\frac{\theta_x}{2}| & \leq \delta/5, \\
	\label{eq:py}|p_y-\cos^2\frac{\theta_y}{2}| & \leq \delta/5, \\
	\label{eq:pz}|p_z-\cos^2\frac{\theta_z}{2}| & \leq \delta/5,
\end{align}
holds with probability at least $(1- 2 e^{-\frac{2}{25}n\delta^2})$, and later show that Equations \eqref{eq:px}, \eqref{eq:py}, and \eqref{eq:pz} imply that $d(u,v) \leq \delta$.

We know in the ideal case, when $n\rightarrow\infty$ the relative frequency $p_x \rightarrow \cos^2\frac{\theta_x}{2}$ but in \pname{2ED} $n$ is finite. So, using Hoeffding's inequality we get,
\begin{align}
	\Pr\left(\abs{p_x-\cos^2\frac{\theta_x}{2}}>\frac{\delta}{5}\right) 
		\leq 2 \exp\left(-\frac{2n^2\delta^2}{25n}\right),
\end{align}
hence Conditions \eqref{eq:px}, \eqref{eq:py}, and \eqref{eq:pz} are all satisfied with probability at least $\left(1-2e^{\left(-2 n\delta^2/25\right)}\right)^3$.
Denoting the vector $u$ in the receiver's basis by $(x_u,y_u,z_u)$, we have
\begin{align}
	x_u = \cos \theta_x = 2 \cos^2\frac{\theta_x}{2} - 1.
\end{align}
So,
\begin{align}
\abs{x-x_u} &= \abs{(2p_x-1) - \left(2 \cos^2\frac{\theta_x}{2} - 1\right)},\\
&= 2\abs{\left(p_x-\cos^2\frac{\theta_x}{2}\right)},\\
\label{eq:pupx}&\leq \ 2\delta/5.
\end{align}
Here, Inequality~\eqref{eq:pupx} follows from Inequality~\eqref{eq:px}.
Similarly we have, 
\begin{align}
	\label{eq:pupyz} y-y_u \leq 2\delta/5 \quad\text{and}\quad z-z_u \leq 2\delta/5.
\end{align}
Using~\eqref{eq:pupx} and \eqref{eq:pupyz}, we get, 
\begin{align}
	d((x,y,z),u)
		&= \sqrt{(x-x_u)^2+(y-y_u)^2+(z-z_u)^2}, \nonumber \\
		&\leq  \sqrt{(2\delta/5)^2+(2\delta/5)^2+(2\delta/5)^2}, \\
		&= \frac{2\sqrt{3}\delta}{5}. 
\end{align}
This means that $(x,y,z)$ is within a sphere of radius $\frac{2\sqrt{3}\delta}{5}$ centered in $u$, so its angle $\theta$ with $u$ is at most $\arcsin(2\sqrt{3}\delta/5)$.
Since $v$ is the normalization of $(x, y, z)$, its angle with $u$ is also $\theta$ and from a simple trigonometric observation, we have,
\begin{align}
	\label{eq:arcdist}d(u,v) = 2 \sin (\theta/2) \leq 2\sin \left(\frac{1}{2} \arcsin(2\sqrt{3}\delta/5)\right).
\end{align}
Moreover, one can check that for all $\alpha \in [0,1],\   \sin \left(\frac{1}{2}\arcsin(\alpha)\right)\leq\frac{5}{4\sqrt{3}}\alpha$, thus,
\begin{align}
	d(u,v) \leq \delta.
\end{align}

So far we have considered only a noiseless channel, let us now turn to the case of a depolarizing channel: if the sender sends a pure state $\ket{\psi}$, the receiver gets the mixed state 
\begin{align}
\label{eq:depolar}\rho = (1-\eps)\proj{\psi} + \eps \frac{\I}{2}.
\end{align}

From Equation~\eqref{eq:depolar} one can see that the effective relative frequency $p_x$ is given by
\begin{align}
p_x &= (1-\eps) p_x' + \frac{\eps}{2},
\end{align}
where $p_x'$ is the relative frequency that the receiver would have got if the channel was noise-free, meaning that $\abs{p'_x-\cos^ 2 \frac{\theta_x}{2}}\leq \delta/5$. Therefore,
\begin{align}
|p_x-\cos^2\frac{\theta_x}{2}| &= |(1-\eps) p_x' + \frac{\eps}{2} -\cos^2\frac{\theta_x}{2} |,\\
\label{eq:upx}&\leq|(1-\eps)\frac{\delta}{5}+\frac{\eps}{2} - \eps\cos^2\frac{\theta_x}{2}|,\\
\label{eq:cosr}&\leq|(1-\eps)\frac{\delta}{5}+\frac{\eps}{2}|,\\
&=(1-\eps)\frac{\delta}{5}+\frac{\eps}{2}.
\end{align}
Here  Inequality~\eqref{eq:cosr} follows because $\eps\cos^2({\theta_x}/{2})$ is positive.

The rest of the analysis remains the same as the noise-free case by replacing $\delta/5$ by $\arcsin(2\sqrt{3}\delta/5)$ in Equation~\eqref{eq:px}.

\end{proof}
 
\section{Step 1: Weak Consensus}

Let us start by giving a more formal definition of a weak consensus protocol.
\begin{definition}
	A $(\delta,\eta)$-weak consensus protocol is a $m$-party protocol, in which each node $P_i$ has an input direction $w_i$ and outputs either a direction $u_i$ or $\perp$, that satisfies the following two properties:
	\begin{description}
		\item[$\boldsymbol{\delta}$-weak persistency]  If there exists a direction $s$ such that for every correct node $P_i$, $d(s,w_i) \leq \delta$, then every correct node $P_i$ outputs a direction $u_i$ with $d(s,u_i) \leq \delta$.
		\item[$\boldsymbol{\eta}$-weak consistency] For every pair of correct nodes $P_i$ and $P_j$ who output $u_i\neq \perp$ and $u_j\neq  \perp$ respectively, we have $d(u_i,u_j)\leq \eta$.
	\end{description}
\end{definition}

\setcounter{algocf}{2}
\begin{algorithm}
\SetAlgorithmName{Protocol}{protocol}{List of Protocols}
	\DontPrintSemicolon
	\LinesNumbered
	\SetKwInOut{Input}{Input}
	\SetKwInOut{Output}{Output}
	\Input{Direction $w_i$}
	\Output{Direction $u_i$ or $\perp$}

Send $w_i$ to all other nodes \;
Receive $a_i[j] \leftarrow$ direction received from $P_j$ \;
Create the set $S_i \leftarrow \{ P_j: d(w_i, a_i[j]) \leq 3\delta \}$ \;
\uIf{$|S_i| \geq m-t$}
	{Assign $u_i\leftarrow w_i$}
\Else{Assign $u_i \leftarrow \perp$}
Output $u_i$
\caption{\pname{Weak-Consensus}}
\end{algorithm}

\begin{theorem} \label{th:weak}
Using a two-party $\delta$-estimate direction protocol that succeeds with probability $\qsucc$, the protocol \pname{Weak Consensus} is a $(\delta, 8\delta)$-weak consensus protocol tolerant to $t < m/3$ faulty nodes that succeeds with probability at least $\qsucc^{m^2-m}$. 
\end{theorem}

\begin{proof}
	After line 2, the property
	\begin{align}
		\label{eq:aproperty}	\forall \text{ correct nodess } P_i, P_j,  \quad d(a_i[j],w_j)\leq \delta,
	\end{align}
	holds with probability at least $\qsucc^{m^2-m}$ since each of the $m$ nodes uses \pname{2ED} $m-1$ times. The rest of the proof shows that Property \eqref{eq:aproperty} implies $\delta$-weak persistency and $8\delta$-weak consistency. This means that \pname{Weak-Consensus} succeeds with probability at least $\qsucc^{m^2-m}$.
	
\paragraph{Weak persistency.}
We assume there exists a direction $s$ such that the input $w_i$ of every correct node $P_i$ satisfies $d(s,w_i)\leq \delta$. Let $P_i$ be a correct node. We now show that $d(s,u_i)\leq \delta$. The idea is to show that $|S_i| \geq m-t$, hence $d(s,u_i) = d(s,w_i) \leq \delta$. This is done by showing that every correct node is in the set $S_i$. Indeed, let us consider a correct node $P_j$, then by triangular inequality we get,
\begin{align}
	d(w_i,a_i[j])  &\leq d(w_i,s)+d(s,w_j)+d(w_j,a_i[j]).
\end{align}
Each of the first two terms is at most $\delta$ by assumption, and the last one is also at most $\delta$ by Property~\eqref{eq:aproperty}. Thus,
\begin{align} 
	d(w_i,a_i[j]) \leq 3\delta.
\end{align}	  
Since there are at least $(m-t)$ non faulty nodes, $|S_i|\geq (m-t)$. This completes the proof of the $\delta$-weak persistency.
	
\paragraph{Weak consistency.}
Let us consider two correct nodes $P_i$ and $P_j$ which output $u_i\neq \perp$ and $u_j \neq \perp$ respectively.
Now we show that $d(u_i,u_j) \leq 8\delta$. The idea is to show that there exists a direction $w_\alpha$ such that $d(u_i, w_\alpha) \leq 4\delta$ and $d(u_j,w_\alpha) \leq 4\delta$. This is done by first showing that there exists one correct node $P_\alpha$ in both sets $S_i$ and $S_j$.
	
For that, let us define the sets $C_i$ and $C_j$ by, 
\begin{align}
	C_i = \{P_l: P_l\in S_i \text{ and node } P_l \text{ is correct} \}, \\
	C_j = \{P_l: P_l\in S_j \text{ and node } P_l \text{ is correct} \}.
\end{align}
We need to prove that $C_i \cap C_j \neq \emptyset$. We do it by contradiction: let us assume that
\begin{align}
	\label{eq:wcon}C_i \cap C_j = \emptyset.
\end{align}
Note that,
\begin{align}
	|S_j| \geq m-t &\Rightarrow |S_j - C_j| + |C_j| \geq m-t,\\
	\label{eq:tsub}&\Rightarrow t + |C_j| \geq m-t,\\
	&\Rightarrow |C_j| \geq m-2t,\\
	\label{eq:tm32}&\Rightarrow  |C_j| > \frac{m}{3}.
\end{align}
Inequality~\eqref{eq:tsub} follows because there can be at most $t$ faulty nodes, and Inequality~\eqref{eq:tm32} since $t<\frac{m}{3}$.
Now,
\begin{align}
	|S_i \cup S_j| 
		&= |(S_i-C_i) \cup (S_j-C_j)\cup C_i \cup C_j|, \\ \label{eq:wconuse}
		&=  |(S_i-C_i) \cup (S_j-C_j)|+ |C_i| + |C_j|,\\
		&\geq |(S_i-C_i)|+ |C_i| + |C_j|,\\
		&= |(S_i-C_i) \cup C_i| + |C_j|,\\
		&= |S_i| + |C_j|,\\
		&\geq (m-t) + |C_j|,\\ \label{eq:wcf}
		&> m - \frac{m}{3} + \frac{m}{3}.
\end{align}	
Here, Equation~\eqref{eq:wconuse} follows from Equation~\eqref{eq:wcon}, and Inequality~\eqref{eq:wcf} from Inequality~\eqref{eq:tm32}. We just proved that $|S_i \cup S_j|>m$ which contradicts the fact that there are exactly $m$ nodes. So, we have $C_i \cap C_j \neq \emptyset$. 

Consider a correct node $P_\alpha\in (C_i \cap C_j )$. We have:
\begin{align}
	d(u_i, w_\alpha)
		& = d(w_i, w_\alpha), \\
		& \leq d(w_i, a_i[\alpha]) + d(a_i[\alpha], w_\alpha), \\
		& \leq 3\delta + \delta. \label{eqn:4deltai}
\end{align}
The factor $3\delta$ comes from the fact that $P_\alpha$ is in $S_i$ and the remaining $\delta$ since $P_\alpha$ is correct. We can do the same reasoning with the node $P_j$, hence we also have:
\begin{align}
	d(u_j, w_\alpha) \leq 4\delta. \label{eqn:4deltaj}
\end{align}

By combining Equations~\eqref{eqn:4deltai} and~\eqref{eqn:4deltaj}, we prove the $8\delta$-weak consistency:
\begin{align}
d(u_i,u_j) \leq d(u_i,w_k) + d(w_k,u_j) \leq 4\delta + 4\delta = 8\delta.
\end{align}
\end{proof}

\section{Step 2: Graded Consensus}

Again, we shall start by giving a formal definition of a graded consensus protocol.

\begin{definition}
A $(\delta,\eta)$-graded consensus protocol is an $m$-party protocol, in which each node $P_i$ has an input direction $w_i$ and outputs a direction $v_i$  as well as a grade $g_i\in\{0,1\}$, that satisfies the following properties:
\begin{description}
	\item[$\boldsymbol\delta$-graded persistency] If there exists a direction $s$ such that for every correct node $P_i$, $d(s,w_i) \leq \delta$, then every correct node $P_i$ outputs a direction $v_i$ such that $d(s,v_i)\leq \delta$ and $g_i=1$;
	\item[$\boldsymbol\eta$-graded consistency] If there exists a correct node $P_i$ who outputs grade $g_i=1$, then for all pairs ($P_i,P_j$) of correct nodes,  $d(v_j,v_k)\leq \eta$.
\end{description}
\end{definition}

\begin{algorithm}
\SetAlgorithmName{Protocol}{protocol}{List of Protocols}
	\LinesNumbered
	\DontPrintSemicolon
	\SetKwInOut{Input}{Input}\SetKwInOut{Output}{Output}
	\Input{A direction $w_i$}
	\Output{A direction $v_i$ and a grade $g_i \in \{0,1\}$}
	
	Run \pname{Weak-Consensus}($w_i$)\; 
	\tcp{This initialises the variables $u_i$ and $a_i[j]$'s}
	\uIf{$u_i = \perp$}{Send flag $f_i = 0$ to all other nodes}
	\Else{Send flag $f_i = 1$ to all other nodes}
	\ForAll{nodes $P_j$}{$f_i[j] \leftarrow$ Receive $f_j$}
	\ForAll{nodes $P_j$ with $f_i[j]=1$}{Create set $T_i[j] \leftarrow \{P_k: f_i[k]=1,$ and $d(a_i[j], a_i[k]) \leq 10\delta\}$ }
	Assign $l_i\leftarrow \arg \max \{|T_i[j]|\}$ \;
	\uIf{$f_i  = 1$}{Assign $v_i \leftarrow w_i$}
	\Else{Assign $v_i \leftarrow a_i[l_i]$}
	\uIf{$|T_i[l_i]| > m - t $}{Assign $g_i \leftarrow 1$}
	\Else{Assign $g_i \leftarrow 0$}
	Output $(v_i, g_i)$

\caption{\pname{Graded-Consensus}}
\end{algorithm}

From Line 2 to Line 7, the nodes send and receive classical bits, there is no approximation here. An important consequence is that $f_i[j] = f_j$ whenever the nodes $P_i$ and $P_j$ are correct.

\begin{theorem}
	Consider that \pname{Weak Consensus} uses a $\delta$-estimate direction protocol that succeeds with probability $\qsucc$. 
Protocol \pname{Graded Consensus} is a $(\delta, 30\delta)$-graded consensus protocol tolerant to $t < m/3$ faulty nodes that succeeds with probability at least $\qsucc^{m^2-m}$.
\end{theorem}

\begin{proof}
Similarly to the \pname{Weak Consensus} protocol, with probability at least $\qsucc^{m^2-m}$, the following property holds:
\begin{align}
	\forall \text{ correct nodes } P_i, P_j,\quad d(a_i[j],w_j)\leq \delta.
\end{align}

\paragraph{Graded persistency.}
We assume there exists a direction $s$ such that, for each correct node $P_i$, $d(s, w_i)\leq \delta$. We first show that every correct node $P_i$ outputs grade $g_i=1$, and then show their output $v_i$ satisfies $d(s,v_i) \leq \delta$.

Let us consider a correct node $P_i$. It outputs $g_i = 1$ if and only if $|T_i[l_i]| \geq m-t$. To show that the later condition holds, we first show that for each of the $(m-t)$ correct nodes $P_j$ we $\abs{T_i[j]} \geq m-t$. Therefore, by definition of $l_i$, we have $\abs{T_i[l_i]}\geq m-t$. Then we show that for all correct nodes $P_\alpha$, we have $d(a_i[j],a_i[\alpha]) \leq 4 \delta$. This will imply that $P_\alpha \in T_i[j]$, and will prove the first part.

Since the nodes $P_j$ and $P_\alpha$ are both correct, and \pname{Weak Consensus} is $\delta$-weak persistent, we know that  $u_j \neq \perp$, $u_\alpha \neq \perp$ with
\begin{align}
	\label{eqn:deltaWeakInGraded}
	d(s,u_j) \leq \delta \quad \text{and}\quad d(s,u_\alpha) \leq \delta.
\end{align}
As a consequence $f_i[j] = f_i[\alpha] = 1$. We also know that $a_i[j]$ and $a_i[\alpha]$ are $\delta$-approximations of $u_j$ and $u_\alpha$ respectively, that is,
\begin{align}
	\label{eqn:deltaApproxInGraded}
	d(a_i[j],u_j)  \leq \delta \quad \text {and}\quad d(a_i[\alpha],u_\alpha)  \leq \delta.
\end{align}
Using the triangular inequality again with the Inequalities~(\ref{eqn:deltaWeakInGraded}) and~(\ref{eqn:deltaApproxInGraded}), we get, 
\begin{align}
	d(a_i[j],a_i[\alpha]) 
	& \leq 
		d(a_i[j], u_j) + d(u_j,s) \nonumber \\
	&\quad + d(s,u_\alpha) + d(u_\alpha,a_i[\alpha]), \\
	&\leq 4\delta.
\end{align}
Since $f_i[j] =1$, the set $T_i[j]$ exists,  and since $f_i[\alpha] = 1$ and $d(a_i[j],a_i[\alpha]) \leq 4\delta \leq 10\delta$, $P_\alpha\in T_i[j]$. This proves that $g_i = 1$.

Now, let us show that $d(s,v_i)\leq \delta$. By $\delta$-weak persistency, we know that $u_i\neq \perp$, therefore, $f_i=1$. In this case, Line 12 assigns $v_i \leftarrow w_i$. As a direct consequence, we get, $d(s, v_i) = d( s, w_i) \leq \delta$. This concludes the proof of the $\delta$-graded persistency.

\paragraph{Graded consistency.}
Let us assume that there exists a correct node that outputs grade $1$. In this case we show that for any two correct nodes $P_i$ and $P_j$ $d(v_i, v_j) \leq 30 \delta$. 

This proof is in three steps. First, we will show that all the correct nodes who are in the sets created at Line 9 are close to each other. More precisely, we will show that for all the correct nodes $P_\alpha$ and $P_\beta$ with $f_\alpha = f_\beta = 1$, we have $d(u_\alpha, u_\beta) \leq 8\delta$. The second step shows that $v_i$ and $v_j$ are $11\delta$-close to some $u_\alpha$  where $P_\alpha$ and $P_\beta$ are correct nodes with $f_\alpha = f_\beta = 1$. The last step combines this two facts to conclude the proof.

Step 1)\quad This first step is a consequence of the $8\delta$-weak consistency of the \pname{Weak Consensus} protocol used at Line 1. Indeed, consider two correct nodes $P_\alpha$ and $P_\beta$ such that $f_\alpha = f_\beta = 1$. This means that $u_\alpha \neq \perp$ and $u_\beta \neq \perp$, hence they satisfy 
\begin{align}
	\label{eqn:weakConsistency}
	d(u_\alpha,u_\beta) \leq 8\delta.
\end{align}

Step 2)\quad We now prove that there exists a correct node $P_\alpha$ such that $d(v_i, u_\alpha) \leq 11\delta$. There are two cases to consider here. First $f_i = 1$: in this case, the correct node $P_i$ outputs $v_i = u_i$, thus $d(v_i, u_i) = 0 \leq 11\delta$. The more interesting case is $f_i = 0$. We are going to show that in this case, there exists a correct node $P_\alpha \in T_i[l_i]$. This is done by showing that the number of nodes in the set $T_i[l_i]$ is more than the number of faulty nodes, that is, $|T_i[l_i]| > m/3$. In a similar manner then for the graded persistency, we will in fact prove that for every correct node $P_k$ with $f_k=1$, $|T_i[k]| > m/3$.

Let us then consider a correct node $P_k$ with $f_k = 1$. By Equation~\eqref{eqn:weakConsistency}, we have $d(u_k, u_{k'})\leq 8\delta$ for every correct node $P_{k'}$ with $f_{k'} = 1$. As a consequence, we also have
\begin{align}
	d(a_i[k],a_i[k'])
		& \leq d(a_i[k], u_k) + d(u_k, u_{k'}) + d(u_{k'}, a_i[k']), \nonumber \\
		& \leq \delta + 8\delta + \delta.
\end{align}
This with Line 9 implies that the set $T_i[k]$ contains every correct node $P_{k'}$ that has $f_{k'} = 1$. 
Let us argue that there are more than $m/3$ such correct nodes. We are in the case where $g_i=1$, that is, $|T_i[l_i]|\geq (m-t)$. We also know that there are at most $t< \frac{m}{3}$ faulty nodes. So, there must be at least $m-2t\geq\frac{m}{3}$ correct nodes in $T_i[l_i]$, that is, there are more than $m/3$ correct nodes $P_{k'}$ with $f_{k'} = 1$.

We just proved that there exists at least one correct node $P_\alpha$ in $T_i[l_i]$, therefore,
\begin{align}
	d(v_i,u_\alpha) 
		& = d(a_i[l_i],u_\alpha), \\
		& \leq d(a_i[l_i], a_i[\alpha]) + d(a_i[\alpha], u_\alpha), \\
		& \leq 10 \delta + \delta. \label{eqn:11deltaj}
\end{align}

Using similar arguments, there exists at least one correct node $P_\beta$ such that 
\begin{align}
	d(v_j, u_\beta) \leq 11\delta. \label{eqn:11deltak}
\end{align}

Step 3)\quad Now using triangular inequality with Inequalities~\eqref{eqn:11deltaj}, \eqref{eqn:weakConsistency}, and~\eqref{eqn:11deltak} we get,
\begin{align}
	d(v_i, v_j) 
		&\leq d(v_i, u_\alpha) + d(u_\alpha, u_\beta) + d(u_\beta, v_j), \\
		&\leq 11\delta + 8\delta + 11\delta. 
\end{align}
This proves the $(30\delta)$-graded consistency of the protocol.
\end{proof}

\section{Step 3: King Consensus}
\label{sec:king}

\begin{definition}
	A $(\delta,\eta)$-king consensus protocol is an $m$-party protocol in which one node $P_k$, called the king, choses a direction $w_k$ and each of the other nodes $P_i$ outputs either a direction $v_i$ or each of them outputs $\perp$, which satisfies the following two properties:
	\begin{description}
		\item[$\boldsymbol\delta$-persistency]  If the king is correct, then all the correct nodes $P_i$ output $v_i \neq \perp$ with $d(w_k,v_i)\leq \delta$.
		\item[$\boldsymbol\eta$-consistency] All correct nodes reach a consensus, that is, they either all output $\perp$, or they all output directions that are $\eta$-close to each other, i.e., for all correct nodes $P_i$ and $P_j$, the distance $d(v_i,v_j)\leq\eta$.
	\end{description}
\end{definition}

Our protocol to solve the king consensus problem uses \pname{Graded-Consensus} and \pname{Classical-Consensus} as subroutines. The latter is a protocol between $m$ parties, in which each node starts with an input bit $g_i$ and outputs a bit $y_i$, that satisfies the following two properties:
	\begin{description}
		\item[Agreement] All correct nodes should output the same bit;
		\item[Validity] If all correct nodes start with the same input $g_i = b$, they should all output this value, that is $y_i = b$.
	\end{description}
\pname{Classical-Consensus} is tolerant to $t < m/3$ faulty nodes (for a protocol see, e.g.,~\cite{PSL80}).

\begin{algorithm}
\LinesNumbered
\SetAlgorithmName{Protocol}{protocol}{List of Protocols}
	\DontPrintSemicolon
	\SetKwInOut{Input}{Input}
	\SetKwInOut{Output}{Output}
	\Input{Id of the king $P_k$.}
	\Output{A direction $v_i$ or $\perp$}

\uIf{I am the king}{
Fix an arbitrary direction $w_k$ \;
Send $w_k$  to all other nodes}
\Else{Receive $w_i \leftarrow$ direction received from the king}
Assign $(v_i, g_i) \leftarrow \pname{Graded-Consensus}(w_i)$\;
Assign $y_i \leftarrow \pname{Classical-Consensus}(g_i)$\; 
\uIf{$y_i = 1$}{Output $v_i$}
\Else{Output $\perp$}
\caption{\pname{King-Consensus}}
\end{algorithm}

\begin{theorem}
Using a $\delta$-estimate direction protocol that succeeds with probability $\qsucc$, \pname{King-Consensus} is a $(\delta,30\delta)$-king consensus protocol that succeeds with probability at least $\qsucc^{m^2}$.
\end{theorem}

\begin{proof}\ 
\paragraph{Persistency.} Let us assume that the king is correct. We want to show that every correct node $P_i$ outputs $v_i \neq \perp$ with  $d(w_k,v_i) \leq \delta$. Since the king is non faulty, with probability at least $\qsucc^m$, we have that for all correct players $P_i$, the distance $d(w_k, w_i) \leq \delta$.

From the $\delta$-graded persistency of \pname{Graded-Consensus} used in Line 6, we know that for all correct nodes $P_i$, $d(v_i,w_k)\leq\delta$ and $g_i=1$ with success proability at least $\qsucc^{m^2}$; And from the validity of \pname{Classical-Consensus}, we have that $y_i = 1$ for all correct nodes $P_i$. Hence all the correct nodes output a $\delta$-approximation of $w_k$ with probability at least $\qsucc^{m^2}$.

\paragraph{Consistency.} To prove consistency we will show that all the correct nodes output $\perp$, or they all output a direction. In this case we also have to show that for every pair $(P_i, P_j)$ of correct nodes, $d(v_i,v_j) \leq 30\delta$.

Since the variables $y_i$ are outputs of \pname{Classical-Consensus}, the agreement property ensures that there exists a bit $b$ such that for all the correct nodes $P_i$, $y_i = b$. 

If $b = 0$, all the correct nodes output $\perp$.

If $b =1$, by validity of \pname{Classical-Consensus}, at least one of the correct nodes, let us denote it by $P_i$, has flag $g_i = 1$. Recall that the ($30\delta$)-graded consistency of \pname{Graded-Consensus} says that we have in this case $d(v_i,v_j)\leq 30\delta$ for every correct nodes $P_i$ and $P_j$.
\end{proof}

\end{document}